\documentclass{IEEEtran}

\usepackage{algpseudocode}
\usepackage{algorithm}
\usepackage{amsmath,amsthm,amsfonts}
\usepackage{graphicx}
\usepackage{subfig}
\usepackage{bbm}
\usepackage{bm}
\usepackage{overpic}
\usepackage{epstopdf}
\usepackage{tikz}
\usepackage{overpic} 
\usepackage{url}

\usetikzlibrary{automata,arrows,positioning,calc}
\definecolor{mycolor}{RGB}{91,155,213}

\newtheorem{lemma}{Lemma}
\newtheorem{theorem}{Theorem}

\DeclareMathOperator*{\argmax}{arg\,max}

\begin{document}

\title{A Distributed Algorithm for Throughput Optimal Routing in Overlay Networks}
\author{Anurag Rai, Rahul Singh, Eytan Modiano
\thanks{Anurag Rai (rai@mit.edu), Rahul Singh (rsingh12@mit.edu) and Eytan Modiano (modiano@mit.edu) are with the Laboratory of Information and Decision Systems (LIDS), Massachusetts Institute of Technology, Cambridge, MA 02139, USA. }
}

\maketitle

\begin{abstract}
We address the problem of optimal routing in overlay networks. An overlay network is constructed by adding new overlay nodes on top of a legacy network. The overlay nodes are capable of implementing any dynamic routing policy, however, the legacy underlay has a fixed, single path routing scheme and uses a simple work-conserving forwarding policy. Moreover, the underlay routes are pre-determined and unknown to the overlay network. The overlay network can increase the achievable throughput of the underlay by using multiple routes, which consist of direct routes and indirect routes through other overlay nodes. We develop a throughput optimal dynamic routing algorithm for such overlay networks called the Optimal Overlay Routing Policy (OORP). 

OORP is derived using the classical dual subgradient descent method, and it can be implemented in a distributed manner. We show that the underlay queue-lengths can be used as a substitute for the dual variables. We also propose various schemes to gather the information about the underlay that is required by OORP and compare their performance via extensive simulations.
\end{abstract}

\section{Introduction}

Optimal routing algorithms\footnote{A routing algorithm is throughput optimal if it can stabilize any traffic that can be stabilized by some routing algorithm.}
have received a significant amount of attention in the literature for the past two decades (e.g. \cite{tassiulas, aurbuch, neely1, lin_shroff}), however, they have had limited success in terms of implementations. One of the main reasons behind the lack of traction is that these policies require additional functionalities that are not supported by the legacy devices. For example, most of these algorithms need the network to be composed of  homogeneous nodes that possess the ability to implement a dynamic routing policy. In contrast, many legacy networks use a single path routing scheme with a work-conserving forwarding policy such as FIFO, and hence can support only a fraction of the achievable throughput. Thus, an implementation of a throughput optimal scheme usually requires a complete overhaul of the network. An overlay architecture for a gradual move towards optimal routing was proposed in \cite{nathan}. This architecture integrates overlay nodes capable of dynamic routing into an underlay network of legacy devices (see Figure \ref{fig: overlay_underlay_network} for an example). In this paper, we develop a throughput optimal dynamic routing algorithm for such overlay networks.

Overlay networks have been used to improve the performance and capabilities of computer networks for a long time. The Internet itself started as a data network built on top of the telephone network. An overlay architecture to improve the robustness of the Iinternet was proposed in \cite{ron}, where alternate overlay paths are used to overcome path loss in the underlay network. Placement for the overlay node to improve path diversity was studied in \cite{path_diversity}. Architectures for designing overlay networks that improve different quality of service metrics have been proposed in \cite{qron}, \cite{akamai}. Currently overlay is being used extensively for applications such as content delivery, multicast, etc.

\begin{figure}[t]
\centering
\includegraphics[scale=.6]{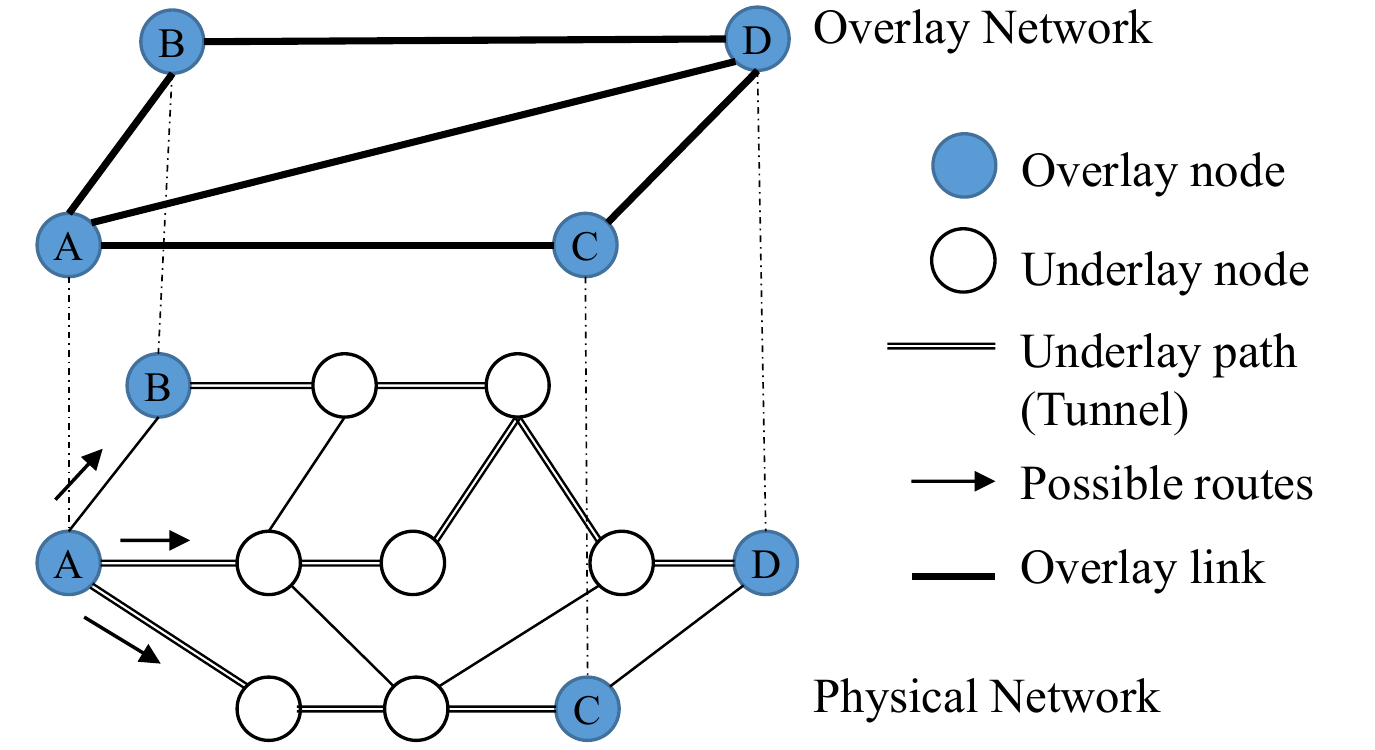}
\caption{Overlay network architecture. If the overlay node A has traffic for node D, it can either route it directly using the tunnel from A to D  or relay it through other overlay node B or C.}
\label{fig: overlay_underlay_network}
\end{figure}

In \cite{nathan}, the authors study the problem of placing the minimum number of overlay nodes into an existing underlay in order to maximize network throughput. In particular, the authors show that with just a few overlay nodes, maximum network throughput can be achieved. However, \cite{nathan} also  
shows that the backpressure routing algorithm of \cite{tassiulas}, which is known to be optimal in a wide range of scenarios, leads to a loss in throughput when used in an overlay network. Then the authors of \cite{nathan} proposes a heuristic for optimal routing called the Overlay Backpressure Policy (OBP). An optimal backpressure like routing algorithm for a special case, where the underlay paths do not overlap with each other, was given in \cite{georgios}. This paper also proposes a threshold based heuristic for general overlay networks. The schemes presented in \cite{nathan} and \cite{georgios} are very similar and were conjectured to be throughput optimal. 

In this paper, we provide a counterexample to show that OBP is in fact not throughput optimal and develop a new optimal routing policy. To derive the optimal policy, we notice that the suboptimality of backpressure arises from its failure to accurately account for congestion in the underlay paths. Traditional backpressure doesn't keep track of the packets in the underlay which can lead the overlay nodes to send too many packets into the underlay creating instability of underlay queues. We will first develop a centralized solution which achieves optimality by limiting the traffic injected into the underlay so that the underlay queues are always bounded. Then we use the intuition gained from this policy to develop a distributed solution which uses the queue backlog information in the underlay to compute the amount of flow transmitted into each underlay path. This policy implicitly favors underlay paths that are less congested and preserves stability of all the queues.



This paper is organized as follows. We describe our model in the next section. In section III, we provide a counterexample to the OBP routing policy. Then in section IV, we provide a centralized stochastic policy that is throughput optimal for overlay networks. In section V, we develop a distributed policy based on the dual subgradient descent method that requires the underlay queue-lengths. In section VI, we propose three approaches to estimating the queue-lengths if they are not available to the overlay nodes. Finally, we verify the performance of our algorithm with extensive simulations. 

\section{Model}
We model the network as a graph $(N, E)$ where $N$ is the set of nodes and $E$ is the set of directed links. The links are capacitated and the capacity of a link $(i,j)\in E$ is given by $c_{ij}$. The nodes can be of two types: underlay or overlay.  We represent the set of all underlay nodes by $U$ and the set of all overlay nodes by $\mathcal{O} = N \backslash U$. The network supports a set of commodities, $K$, where each commodity $k\in K$ is defined by a source-destination pair. For the ease of exposition, we will formulate the problem with all the sources and destinations being overlay nodes. In Section \ref{sec: underlay_sources}, we discuss how the same solution can be applied when this is not the case. The time is slotted and indexed by $t$. We remove the time index for notational simplicity if removing it doesn't create ambiguity.

\subsection{Overlay}
The overlay network consists of the controllable nodes $\mathcal{O}$ which are capable of implementing a dynamic routing algorithm. The links between two overlay nodes can either be a direct edge or a path through the underlay referred to as a {\em tunnel}. A tunnel $l$ is a sequence of nodes $l_1, l_2, \dots, l_{|l|}$ where $|l|$ is the length of the tunnel. We represent the set of all the tunnels in the network by $L$. 


Since a tunnel connects two overlay nodes, $l_1$ and $l_{|l|}$ are overlay nodes, and $l_{|2|}, ..., l_{|l|-1}$ are underlay nodes. When a packet is sent into a tunnel $l$, node $l_1$ encapsulates it into a packet destined to node $l_{|l|}$ and forwards it onto the underlay node $l_2$. The route taken by the tunnel is dictated by the path from $l_2$ to $l_{|l|}$ which is assigned by the underlay. When the packet reaches $l_{|l|}$, it is decapsulated and enqueued at the node. An example of the different type of links in an overlay network is given in Figure \ref{overlay_queues}. This overlay network consists of one direct link (1,4) and three tunnels (1,3,4), (2,3,4) and (2,3,5).

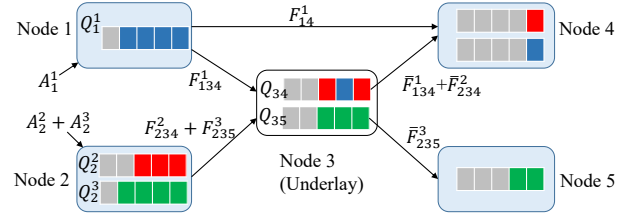
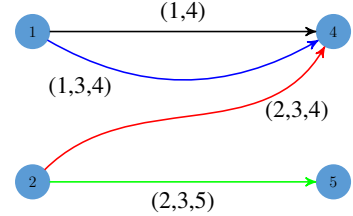
\begin{figure}[ht]
\centering
\begin{tabular}{c}
\subfloat[Example of overlay and underlay queues in the physical network. The variables labelling the edges represent the number of packets transmitted for each commodity on each tunnel.]{
\includegraphics[scale=.4]{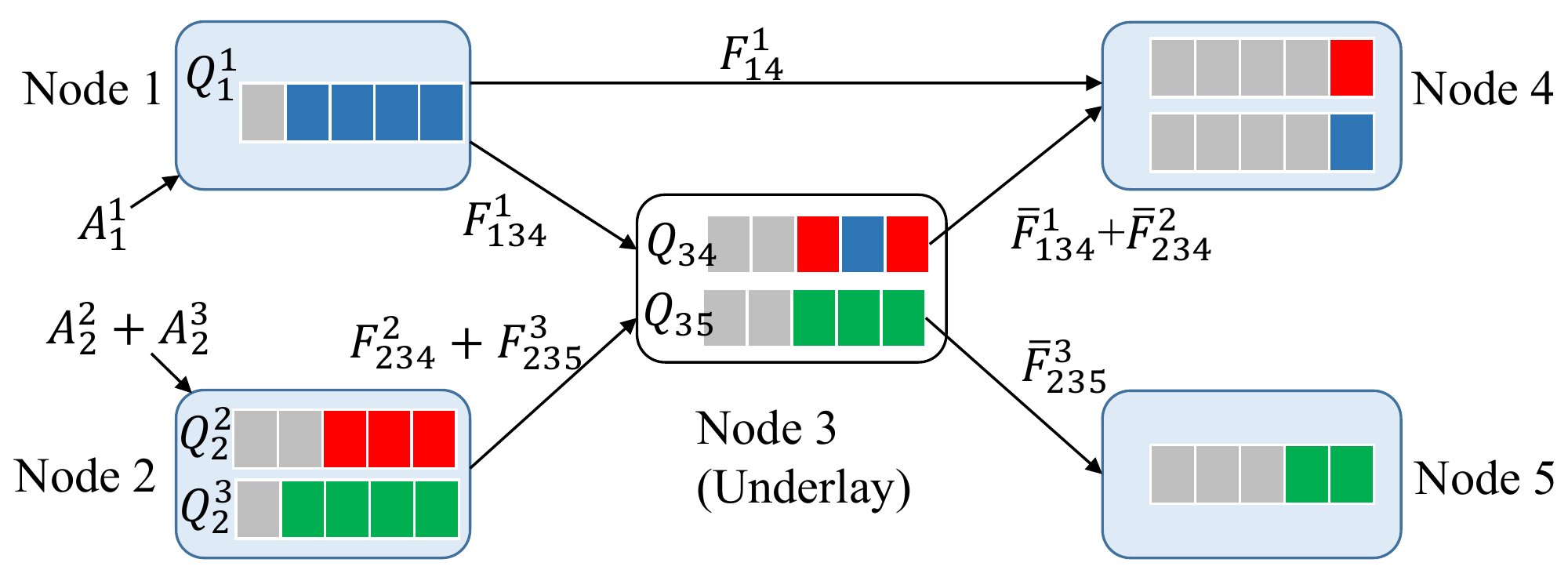}
}\\
\subfloat[Corresponding overlay network. Each tunnel in the overlay network is represented by a sequence of nodes traversed by it.]{
\makebox[7cm][c]{
	\begin{tikzpicture}[-, >=stealth', auto, semithick ]
	
	\tikzstyle{every state}=[thick,text=black,scale=.5,color=mycolor, fill=mycolor,]
	
	\node[state,text=black] (1) at (0,1) {\large $1$};
	\node [state,text=black] (2) at (0,-1) {\large $2$};
	\node[state,text=black] (4) at (4,1) {\large $4$};
	\node[state,text=black] (5) at (4,-1) {\large $5$};

	\draw [->,black](1) -- (4) node [black, midway] {\small (1,4)};
	\draw [->,blue]  (1) to [bend right] (4) node [black, midway] {\small \qquad \qquad (1,3,4)};
	\draw [->,red] (2) to [in=-120] (4) node [black,below=0.6cm] {\small (2,3,4) \qquad \qquad };
	\draw [->,green] (2) -- (5) node [black,midway,below] {\small (2,3,5)};
	\end{tikzpicture}
}
}
\end{tabular}
\caption{A physical network and its corresponding overlay network.}
\label{overlay_queues}
\end{figure}

Each overlay node $i$ maintains a queue for each commodity $k$ and the backlog is represented by $Q^k_i$. The number of external commodity $k$ packets that arrive at node $i$ represented by $A^k_i$. Let $F^k_{l}$ represents the amount of packets injected into the tunnel $l$, and $\bar{F}^k_{l}$ represents the number of packets that exit tunnel $l$. The quantities are different because a packet sent into the tunnel might not exit the tunnel for several time-steps. Let $F^k_{ij}$ represent the number of commodity $k$ packets that are transmitted on an overlay to overlay link $(i,j)$. Figure \ref{overlay_queues} illustrates the meaning of these variables on a simple network. The backlog of commodity $k$ packets at overlay node $i$ evolves as follows:
\begin{align}
Q^k_i(t+1) &= \left[ Q^k_i(t) - \sum_{j\in \mathcal{O}} F^k_{ij}(t)- \sum_{l\in L: l_1 = i} F^k_{l}(t) + \right. \nonumber \\
	&\qquad \qquad\left.  \sum_{j\in \mathcal{O}} F^k_{ji}(t) + \sum_{l\in L:l_{|l|}=i} \bar{F}^k_{l}(t) + A^k_i (t) \right]^+  \nonumber
\end{align}
Here, $\{l\in L: l_1 = i\}$ are all the tunnels that start at node $i$, $\{l\in L:l_{|l|}=i\}$ are the tunnels that end at node $i$, and $[.]^+ = \max(.,0)$. Packets are removed at the destination node, hence the backlog of a commodity at its destination is zero.

We assume that all the traffic arrivals $A^k_i$ are i.i.d. with a mean of $\lambda^k_i$. We also assume that the arrival rate vector $\lambda$ is in the interior of the throughput region of the overlay network $\Lambda$ \cite{nathan}. We will be designing a dynamic routing policy that controls $F^k_{l}$ and $F^k_{ij}$ at each time-step so that both the overlay and the underlay queues stabilize.


\subsection{Underlay}
The underlay network consists of the uncontrollable nodes $U$. These nodes have a static routing policy which assigns a fixed path between each pair of nodes in the underlay. 
The paths are assumed to be acyclic and unique, which ensures that all the tunnels are acyclic an that they take a fixed route through the underlay.

An underlay node maintains a queue per outgoing link. The backlog on the queue associated with the link $(a,b)$ is represented by $Q_{ab}$. The queues have infinite buffer space hence packets are not dropped. When a packet arrives at an underlay node, the node looks up the link assigned to it based on its destination and enqueues it on the corresponding link. Since several tunnels of the overlay network can pass through the same underlay link the underlay queues accumulates packets from several different tunnels and commodities. An example of an underlay queue that is shared by several tunnels is presented in Figure \ref{overlay_queues}. Packets from both the tunnels (1,3,4) and (2,3,4) are queued on the link (3,4).

The underlay employs a work-conserving forwarding scheme that is ``universally stable'' as defined in \cite{andrews}. This assumption ensures that if the number of packets injected into the underlay at each timeslot satisfies the capacity constraints of the tunnels, then the underlay queues are deterministically bounded. Specifically, under a universally stable forwarding policy, an underlay queue corresponding to link $(a,b)$ is always deterministically bounded if
\begin{equation}
\sum_{l\in L:(a,b)\in l} \sum_k F^k_l(t) <c_{ab} \forall t.  \label{eq: tunnel_cap}
\end{equation}
Here $\{l\in L:(a,b)\in l\}$ is the set of tunnels that pass through the link $(a,b)$. We refer to such constraints as the {\em tunnel capacity constraints}. Several work-conserving policies that are universally stable are given in \cite{andrews}.


\section{Background}
The problem of optimal routing in an overlay network was first studied in \cite{nathan}, where it was shown that backpressure routing, which is known to be throughput optimal in a range of scenarios, is not optimal for overlay networks, and proposed a heuristic called the Overlay Backpressure Policy (OBP). The OBP heuristic was conjectured to be throughput optimal.

For each tunnel $l$ and commodity $k$ OBP keeps track of the {\em packets in flight} $H^k_l$, which is the number of packets that have been transmitted into the tunnel by node $l_1$ but haven't reached node $l_{|l|}$. The weight for each commodity over the tunnel $W^k_{l}(t)$  is computed as follows $$W^k_{l}(t) = Q^k_{l_1}(t) - H_{l}^k (t)- Q_{l_{|l|}}^k(t).$$ 
A link $(i,j)$ that connects two overlay nodes can be thought of as a tunnel $l=(i,j)$ with no underlay node, hence the weight is computed as $$W^k_{l}(t) = Q^k_{l_1}(t) - Q_{l_{|l|}}^k(t).$$ Then, the commodity with the highest weight sends its packets into the tunnel provided that the weight is positive. A precise description of the OBP is given in Algorithm \ref{alg: obp}. 

\begin{algorithm}[h!]
\caption{Overlay Backpressure Policy (OPB):}
\label{alg: obp}
For each tunnel $l$ at each time-step $t$:
\begin{enumerate}
\item Compute the commodity $k^*$ that maximizes the weight $W^k_{l}(t)$, 
$$k^* \in \arg \max_k W_{l}^k(t).$$ 
Ties are broken arbitrarily.
\item Transmit $\mu$ packets into the tunnel where 
$$\mu = \left\{ \begin{array}{l} c_{l_1l_2} \text{ if }W_{l}^{k^*}(t) > 0\\ 0, \text{ otherwise,} \end{array} \right.$$ 
where $c_{l_1l_2}$ is the capacity of the first link of tunnel $l$.
\end{enumerate}
\end{algorithm}

This policy makes sense intuitively because it encourages utilizing the tunnels that have less packets in them. When a tunnel is congested, the number of packets in flight is high, which encourages the overlay nodes to use alternate routes and send packets into the tunnel only when the backlog in the overlay is extremely high. This behavior is common to backpressure-based optimal routing algorithms. Moreover, OBP reduces to backpressure routing when all the nodes are overlay nodes.


We present the following counterexample to show that the OBP is not throughput optimal. Consider a network topology given in Figure \ref{fig: counter_topology} where all the links are unit capacity. There are three commodities with source $s_i$ and destination $d_i$, $i=1,2,3$. The source and the destination are overlay nodes, whereas the nodes 1, 2 and 3 (in gray) are underlay nodes. The underlay nodes use the FIFO queuing discipline\footnote{From \cite{tassiulas93} we know that FIFO is throughput optimal for a ring which is the underlay topology in this example.}. Each commodity in this network has two tunnels to the destination, e.g. $(s_1, 1, 2, d_1)$ and $(s_1, 3, 1, 2, d_1)$. Note that the shorter tunnels do not share any links between them. So, if the shorter tunnel is chosen by each commodity, this network can support an arrival rate vector of $[1,1,1]$. 

\begin{figure}[ht]
\centering
\subfloat[Topology]{
\includegraphics[scale=.6]{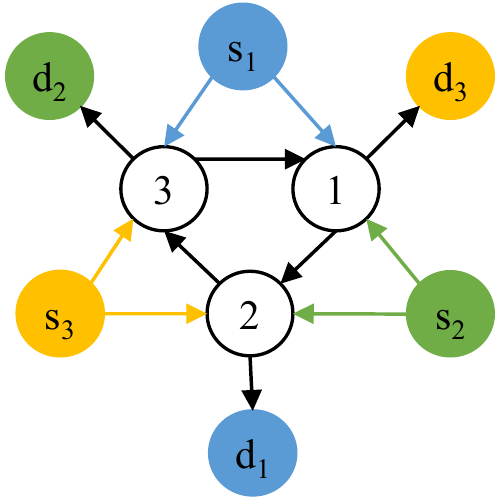}
\label{fig: counter_topology}
}
\subfloat[{Total backlog in the network for arrival rate vector of $[0.8, 0.8, 0.8]$ }] {
\includegraphics[scale=.6, trim=0 12cm 25cm 0]{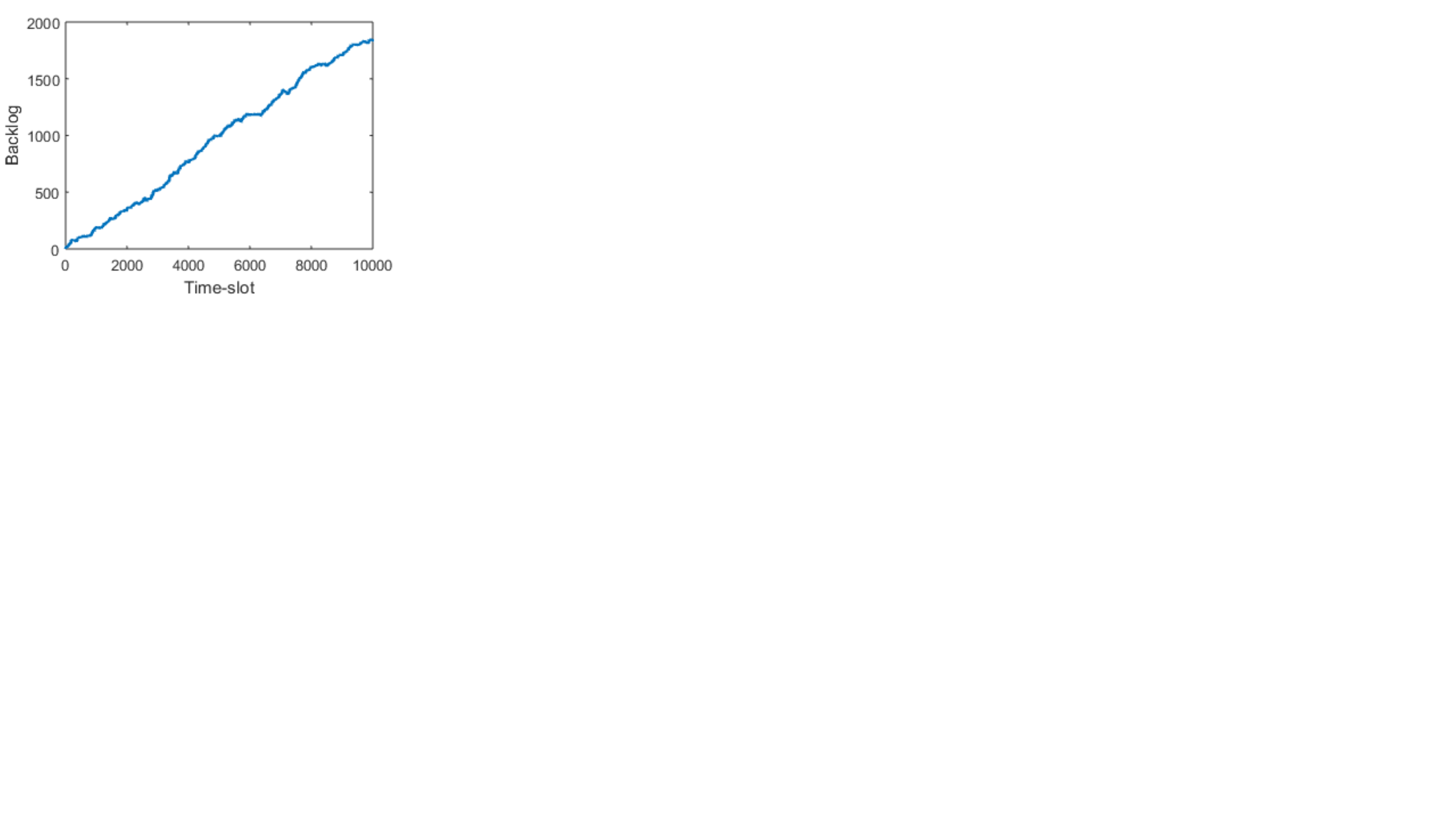}
\label{fig: counter_plot}
}
\caption{Counterexample for throughput optimality of the Overlay Backpressure Policy of \cite{nathan}.}
\end{figure}

Let us consider Poisson arrivals with the rate vector of $[0.8, 0.8, 0.8]$, which is clearly inside the stability region. To support this rate OBP has to send most of its traffic through the shorter tunnels. However, as we show below, congestion can lead traffic to use longer tunnels, which leads to instability.  A simulation result showing this instability is given in Figure \ref{fig: counter_plot}.

This instability is caused by overlapping tunnels where congestion in one tunnel forces commodities to use the longer tunnels which in turn leads to more congestion. Consider the situation where the number of packets in flight is large for tunnel $(s_1,1,2,d_1)$. So, commodity 1 traffic is routed through the tunnel $(s_1,3,1,2,d_1)$. This means that the link (3,1) is being used by commodity 1 packets, which creates congestion for commodity 3 over tunnel $(s_3,3,1,d_3)$ forcing commodity 3 traffic onto tunnel $(s_3,2,3,1,d_3)$. This problem continues for the tunnels used by commodity 2, which in turn create congestion for the tunnels of commodity 1 forcing its traffic onto tunnel $(s_1,3,1,2,d_1)$ further exacerbating the situation. This cyclical nature of increased congestion makes all the commodities unstable.

\section{Centralized solution}
We begin by providing a centralized optimal routing policy for overlay networks. In this section we assume that the underlay topology is known, and a centralized controller can make the routing decisions at the overlay. The key to obtaining a throughput optimal policy is to realize that the underlay cannot make dynamic decisions, hence, the overlay necessarily has to take into account the capacities of the underlay links while making scheduling decisions.

 Our algorithm works by choosing the scheduling decision which minimizes the $T$-slot drift of the quadratic Lyapunov function of the overlay queues \cite{neely_book}. This is similar in spirit to the backpressure routing algorithm, which implements a schedule that minimizes the Lyapunov drift at every slot. In our set-up, multislot drift needs to be considered since packets that are sent into the tunnel take several timeslots to come out of the tunnel.  In addition to minimizing the drift, we also have to make sure that the underlay queues are bounded. Because we assume that the underlay forwarding scheme is universally stable, we are able to guarantee that underlay queues are bounded once the tunnel capacity constraints (\ref{eq: tunnel_cap}) are met. Thus the algorithm seeks a scheduling decision that minimizes the drift subject to the tunnel capacity constraints.
  
To simplify the notation, in this section, a link between two overlay nodes will be viewed as a tunnel which does not comprise of any underlay nodes. We divide the time into $T$-slot duration frames and consider minimizing the $T$-slot drift. At the beginning of each frame, we solve the optimization problem~\eqref{prob} in a centralized fashion. The solution to~\eqref{prob} minimizes the drift of a quadratic Lyapunov function, while simultaneously satisfying the tunnel capacity constraints. The solution gives us ${F^k_l}^*$, which is the number of packets commodity $k$ must send into tunnel $l$ in order to minimize the drift. A complete description of the policy is given in Algorithm \ref{alg: centralized}.

\begin{algorithm}
\caption{Centralized Policy}
\label{alg: centralized}

\begin{enumerate}
\item At the beginning of each frame solve the following optimization problem:
\begin{align}
{F^k_l}^* = \argmax_{F^k_l} &\sum_{k,l}F^k_{l}  [Q^k_{l_1}(t) - Q^k_{l_{|l|}}(t)] \label{prob}\\
s.t. &\sum_{l\in L:(a,b)\in l} \sum_k F^k_{l} \le c_{ab},  \forall (a,b) \in E  \label{eq: tunnel_constraint}\\
 & F^k_l \ge 0 \label{positivecons}
\end{align}

\item Send ${F_l^k}^*$ packets of commodity $k$ into tunnel $l$ each time slot in the frame. If ${F_l^k}^*$ is not an integer, approximate it by sending $p$ packets every $q$ slots so that  $\frac{p}{q} \approx {F_l^k}^*$.
\end{enumerate}

\end{algorithm}

In the special case when a tunnel $l$ does not share any link with other tunnels, we see that ${F_l^k}^\star$ can be computed independently of the other tunnels. The commodity $k^*$ is the one with the highest differential backlog $Q^k_{l_1}(t) - Q^k_{l_{|l|}}(t)$ and $F_l^{k^*}$ is chosen to be the capacity of the smallest link in the tunnel. Thus our algorithm resembles the backpressure routing except for the fact that the packets can face large delays while passing the underlay. However, when there is a shared link, all the tunnels that share the link are required to exchange their backlog information. 


Next we show that the Algorithm~\ref{alg: centralized} stabilizes the network queues if the arrival rate vector $\lambda$ lies in the interor of the stability region. We use $T$-slot Lyapunov drift analysis to prove that these queues are {\em strongly stable} \cite{neely_book}.
\begin{theorem} \label{thm: centralized}
For any arrival rate vector $\lambda$ in the stability region $\Lambda$ and a large enough frame length $T$, the policy given in Algorithm \ref{alg: centralized} stabilizes all the queues in the network.
\end{theorem}
\begin{proof}
The proof of the theorem is given in the Appendix. 
\end{proof}

\section{Fluid Formulation and Distributed Solution}
The centralized policy in the previous section requires the knowledge of the underlay topology which might not be known to the overlay. Moreover, having a centralized controller is often impractical. We now consider the fluid model of the network and develop a decentralized policy. Fluid models have been successfully utilized to establish the stability of queueing networks (e.g. \cite{bramson,dai1995}). 

 Let $f_{ij}^k$ be the flow assigned to commodity $k$ on the link $(i,j)\in E$, and $f_l^k$ be the flow assigned to commodity $k$ on the tunnel $l\in L$. Let $f$ denote the vector containing all the flow variables. The arrival rate of commodity $k$ at overlay node $i$ is represented by $\lambda^k_i$, and we assume that the vector of arrival rates $\lambda$ is in the interior of the stability region. For simplicity, we will assume $\lambda$ to be a constant, however, if it is time-variying, we note that the technical results hold as long as the arrival rate is bounded at each time-step and the expected value $E[\lambda(t)]$ exists. The problem of stabilizing the network queues can be formulated as a linear program that finds a feasible flow allocation on all the links and tunnels,
\begin{align}
\max & \text{ 0 } \nonumber \\
s.t. &\sum_{l:(i,j)\in l} \sum_k f^k_{l} \le c_{ij},  \forall (i,j): i\in U, j\in N  \label{eq: tunnel_constraint_fluid}\\
&\sum_{l:(i,j)\in l} \sum_k f^k_{l} \le c_{ij},  \forall (i,j): i\in \mathcal{O}, j\in U  \label{eq: cap1_fluid}\\
 & \sum_j f_{ij}^k + \sum_{l:l_1=i} f^k_l - \sum_{j} f^k_{ji} - \nonumber \\
		&\qquad  \qquad \qquad \sum_{l:l_{|l|} = i} f_l^k - \lambda^k_i \ge 0, \forall i\in \mathcal{O},k \label{eq: flow_conservation_fluid}\\
 &\sum_k f^k_{ij} \le c_{ij}, \forall i,j \in \mathcal{O} \label{eq: cap2_fluid}\\
 &f^k_{ij}, f^k_l  \ge 0, \label{eq: non-negativity_fluid}
\end{align}
Here, the inequalities (\ref{eq: tunnel_constraint_fluid}) are the tunnel capacity constraints which are the fluid version of (\ref{eq: tunnel_cap}). Each one of these constraints correspond to an uncontrollable link, i.e. a link between two underlay nodes or a link that goes from underlay to an overlay node. Inequalities (\ref{eq: cap1_fluid}) are the link capacity constraints corresponding to the first link in the tunnel, i.e. the links that go from an overlay node to an underlay node. This link is responsible for controlling the rate received by the underlay links. Inequalities (\ref{eq: flow_conservation_fluid}) are the flow conservation constraints on the overlay network. Note that the flow conservation constraints are not required for the underlay because for each tunnel $l$, there is a single route and the flows coming into the underlay are feasible because of (\ref{eq: tunnel_constraint_fluid}). That is, for a tunnel $l$, when $f$ is a feasible solution,
$$f^k_l = f^k_{l_1,l_2} = f^k_{l_2,l_3} = ... = f^k_{l_{|l|-1}, l_{|l|}}.$$
Constraints (\ref{eq: cap2_fluid}) are the capacity constraints for the overlay links.

\subsection{Dual problem}
We now formulate the dual problem so that it can be solved with the subgradient descent method \cite{low, bertsekas_convex}.  Let $q_{ij}$ and $q_i^k$ denote the dual variables for the tunnel constraints (\ref{eq: tunnel_constraint_fluid}) and the flow conservation constraints (\ref{eq: flow_conservation_fluid}) respectively, and let $q$ represent the vector containing all the dual variables. The Lagrangian function is given by, 
\begin{align}
L(f, q ) &= \sum_{(i,j):i\in U}q_{ij} \left( c_{ij} - \sum_{l: (i,j)\in l} \sum_k f^k_{l} \right) + \sum_{i\in \mathcal{O}, k} q^k_i  \nonumber\\
		&\left(\sum_j f_{ij}^k + \sum_{l:l_1=i} f^k_l - \sum_{j} f^k_{ji} - \sum_{l:l_{|l|} = i} f_l^k - \lambda^k_i  \right) \nonumber \\
&= \sum_l \sum_k f^k_l \left(q^k_{l_1}-\sum_{(i,j)\in l:i \in U}q_{ij}-q^k_{l_{|l|}}\right) + \nonumber\\
& \sum_{(i,j)} \sum_k f^k_{ij} ( q^k_{i} - q^k_{j})  + \sum_{(i,j):i\in U} q_{ij}c_{ij} - \sum_{i\in \mathcal{O}, k} q^k_i \lambda^k_i \label{eq: optimal_f},
\end{align}
where the second equality is obtained by rearranging the terms so that the flow variables are factored out instead of the dual variables.

Let $X$ be a set such that any $f \in X$ satisfies the constraints (\ref{eq: cap1_fluid}), (\ref{eq: cap2_fluid}) and (\ref{eq: non-negativity_fluid}). Note that these constraints can be enforced locally by an overlay node using only locally available information. This property will be essential in designing the decentralized algorithm. 
The dual objective function corresponding to the problem~\eqref{eq: tunnel_constraint_fluid} is $$D(q) = \max_{f \in X} L(f,q).$$
The dual problem is given by, 
\begin{align}
\min_q &\quad D(q) \label{dual_problem}\\
\text{s.t.} &\quad q \ge 0. \nonumber
\end{align}
Since the primal problem~\eqref{eq: tunnel_constraint_fluid} is a linear program, the duality gap is zero (Slater's condition~\cite{bertsekas_convex}). Hence, solution of the dual~\eqref{dual_problem} yields a feasible flow allocation.

\subsection{Distributed solution} \label{sec: solution}
The subgradient method works by initializing the dual variables with a value $q(0)\ge 0$, and then iterating on them until it converges to optimal $q^\star$. Each iteration involves computing a subgradient $g$ of $D$ at the current value of the dual variables, then updating the dual variables as follows:
\begin{equation} q(t+1) = \left[q(t) - \alpha(t) g(t) \right]^+.\label{eq: dual_update} \end{equation}
Here $\alpha(t)$ is positive scalar step-size. The dual variables are known to converge to the optimal $q^\star$ if the step-sizes $\alpha(t)$ are chosen appropriately. However, if $\alpha(t)\equiv \alpha$, then the iterates~\eqref{eq: dual_update} converge to a bounded neighbourhood of $q^\star$~\cite{bertsekas_convex}.

Let ${f^k_l}^*$ and ${f^k_{ij}}^*$ be the values of flow variables which maximize the Lagrangian $L(f,q)$ over $f \in X$ for a fixed $q$, i.e. $D(q) = L(f^*,q)$. From \cite{bertsekas_convex} we know that a subgradient of $D(q)$ is given by a vector $g$ with entries as,
\begin{align}
& g_{ij} = c_{ij} - \sum_{l: (i,j) \in l}  \sum_k {f^k_{l}}^*,  \text{ and } \label{eq: subgradient1}\\
&g^k_{i} =\sum_j {f_{ij}^k}^* + \sum_{l:l_1=i} {f^k_l}^* - \sum_{j} {f^k_{ji}}^* - \sum_{l:l_{|l|} = i} {f_l^k}^* - \lambda^k_i. \label{eq: subgradient2}
\end{align}
Now we can use the recursive equation (\ref{eq: dual_update}) to update the dual variables.

The only necessary step that we haven't covered so far is the computation of ${f^k_l}^*$ and ${f^k_{ij}}^*$. A careful observation of equation (\ref{eq: optimal_f}) and the set $X$ shows that this is a simple optimization problem that can be solved in a decentralized fashion. The objective is a weighted sum of the flow variables, and the constraints that form $X$ are the link capacity constraints. At a high level, for each overlay link, the solution chooses the maximum value of the flow variable that corresponds to the commodity with the highest positive weight. A complete algorithm to compute the optimal flow variables and update the dual variables is given in Algorithm \ref{alg: OORP}.

\begin{algorithm}[h!]
\caption{Optimal Overlay Routing Policy (OORP)}
\label{alg: OORP}
At each time-step $t$, overlay node $i$ does the following: 

{\bf Optimal flow variables computation (used to obtain the subgradients):}

An overlay to overlay link $(i,j)$ computes the flow variables ${f^k_{ij}}^*$:
\begin{itemize}
\item Let $k^{opt} \in \argmax_k q^k_i - q^k_j,$ ties are broken arbitrarily. The weight of commodity $k^{opt}$ in this link is $W^{opt}_{ij} = q^{k^{opt}}_i - q^{k^{opt}}_j$.
\item For $k=k^{opt}$, 
    $${f^{k}_{ij}}^* = \left\{ \begin{array}{l}  c_{ij} \text{ if } W^{opt}_{ij} > 0 \\ 0, \text{ otherwise} \end{array} \right.$$
\item For all $k \ne k^{opt}$, ${f^k_{ij}}^* = 0$.
\end{itemize}

Each overlay to underlay link $(i,j)$ computes the flow variable $f^k_l$ for all $l:(l_1,l_2) = (i,j)$:
\begin{itemize}
\item Let 
\begin{equation}
(l^{opt},k^{opt}) \in \argmax_{l:(l_1,l_2) = (i,j),k} q^k_{l_1}-\sum_{(a,b)\in l:a \in U}q_{ab}-q^k_{l_{|l|}}. \label{eq: subgradient}
\end{equation}
Ties are broken arbitrarily. Let the weight of commodity $k^{opt}$ in the tunnel be $$W^{opt}_l = q^{k^{opt}}_{{l^{opt}}_1}-\sum_{(a,b)\in l^*:a \in U}q_{ab}-q^{k^{opt}}_{{l^{opt}}_{|l^{opt}|}}.$$
\item For $(l,k) = (l^{opt},k^{opt})$, $${f^k_l}^* = \left\{ \begin{array}{l} c_{ij} \text{ if } W^{opt}_l> 0 \\ 0, \text{ otherwise} \end{array} \right.$$
\item For all $(l,k):l \ne l^{opt}$ or $k \ne k^{opt}$, ${f^k_l}^* = 0$.
\end{itemize}

{\bf Data transmission: }\\
Transmit ${f_{ij}^k}^*$ amount of commodity $k$ traffic into each overlay to overlay link $(i,j)$ and ${f_{l}^k}^*$ amount of commodity $k$ traffic into each tunnel $l$.
\vspace{2mm}

{\bf Dual variables update:}\\
Performed by an overlay node i:
\begin{align}
q^k_i(t+1) &= \left[q(t) - \alpha(t) \left(\sum_j {f_{ij}^k}^* + \sum_{l:l_1=i} {f^k_l}^* \right. \right. \nonumber\\
		& \qquad \left. \left. - \sum_{j:(j,i)\in E} {f^k_{ji}}^* - \sum_{l:l_{|l|} = i} {f_l^k}^* - \lambda^k_i  \right)  \right]^+ \label{eq: overlay_q_fluid}
\end{align}
Performed by an underlay node i:
\begin{align}
q_{ij}(t+1) &= \left[q(t) - \alpha(t) \left( c_{ij} - \sum_{l: (i,j) \in l}  \sum_k {f^k_{l}}^* \right)  \right]^+ \label{eq: underlay_q_fluid}
\end{align}
\end{algorithm}

\subsection{Queue-lengths as dual variables}
The subgradient descent algorithm presented in the Algorithm \ref{alg: OORP} requires the network to explicitly keep track of the dual variables. In order to implement the algorithm in a decentralized fashion, each underlay node $i$ needs to maintain a dual variable $q_{ij}$ for each link $(i,j)$, and each overlay node $i$ needs to maintain a dual variable $q^k_i$ for each commodity $k$. This is a reasonable assumption for the overlay nodes, but not justified for the uncontrollable underlay. To get around similar problems of not having a dual variable, works such as \cite{low_vegas}, \cite{lin_shroff}, etc. have proposed approximating them with the corresponding queue lengths. The argument behind this procedure is that the subgradients are proportional to the change in queue-lengths, so that the queue-lengths will move in the same direction as the dual variables. Next, we give an example in which this proportionality does not hold. In spite of this issue, we show that the queue-lengths can provide a good approximation for the dual variables.

We first observe that the dual variable update equations (\ref{eq: overlay_q_fluid}) and (\ref{eq: underlay_q_fluid})  are the same as the queue update equations when the flows sent into the tunnels $f_l^k$ are feasible for the underlay, i.e. when no queues buildup in the underlay. But when the flows do not satisfy the tunnel capacity constraints, the underlay queues build up, and the flows get reduced from their initial value as they pass through the bottleneck links. This decrease in the flow size is not captured in these dual variable update equations (\ref{eq: overlay_q_fluid}), (\ref{eq: underlay_q_fluid}). Consider the simple network shown in Figure \ref{fig: not_proportional}. There is one commodity, $k=1$, with source node 1 and destination node 4, and a single tunnel $l=(1,2,3,4)$. Suppose that at a certain iteration, $q_1^1 > q_4^1$, hence ${f_l^k}^* = 3$. This flow into the tunnel gets bottlenecked at link $(2,3)$ so node 3 only receives a flow of 1. In this situation, equation (\ref{eq: underlay_q_fluid}) predicts that the queue-length for $q_{34}$ would increase because a flow of size 3 was sent into the tunnel and the capacity of the link is 2, however this queue can only decrease or stay unchanged at 0.

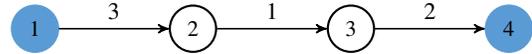
\begin{figure}[ht]
\centering
\begin{tikzpicture}[->, >=stealth', auto, semithick, node distance=3cm]

\tikzstyle{every state}=[thick,text=black,scale=.7]

\node[state, draw=mycolor, fill=mycolor, text=black] (1) {\large 1};
\node [state] (2)[right of=1] {\large 2};
\node[state] (3)[right of=2] {\large 3};
\node[state, draw=mycolor, fill=mycolor, text=black] (4)[right of=3] {\large 4};

\draw (1) -- (2) node [midway]{ \small 3};
\draw (2) -- (3) node [midway]{\small 1};
\draw (3) -- (4) node [midway]{\small 2};

\end{tikzpicture}

\caption{Link $(3,4)$ never builds a queue as the flow gets bottlenecked by (2,3).}
\label{fig: not_proportional}
\end{figure}

To capture this reduction of the flow sizes in the tunnel, we model the queuing in the network as follows:
\begin{align}
\hat{q}^k_i(t+1) &= \left[\hat{q}(t) - \alpha(t) \left(\sum_j {f_{ij}^k}^* + \sum_{l:l_1=i} {f^k_l}^* \right. \right. \nonumber\\
		& \qquad \left. \left. - \sum_{j:(j,i)\in E} {f^k_{ji}}^* - \sum_{l:l_{|l|} = i} \epsilon^k_l(i) {f_l^k}^* - \lambda^k_i  \right)  \right]^+ \label{eq: overlay_approx}\\
\hat{q}_{ij}(t+1) &= \left[\hat{q}(t) - \alpha(t) \left( c_{ij} - \sum_{l: (i,j) \in l}  \sum_k \epsilon^k_l(i,j) {f^k_{l}}^* \right)  \right]^+ \label{eq: underlay_approx}
\end{align}
where $\epsilon^k_l(i), \epsilon^k_l(i,j) \in [0,1]$ represent the reduction suffered by the corresponding flows before arriving at node $i$. These quantities are implicitly determined by the network at each time-step depending on the scheduling policy in the underlay. In the example presented above, for any work conserving scheme, $\epsilon_l^k(3,4) = 1/3$. We will show that for any value of $\epsilon$ in the set $[0,1]$ the queue-lengths will converge to the optimal dual variables. Let $g$ be the true subgradient of $D$ at $q$, and $\hat{g}$ be the approximate subgradient after the reduction, then we can represent the queuing equation as $$\hat{q}(t+1) = \left[\hat{q}(t) - \alpha(t) \hat{g}(t) \right]^+,$$ and $\hat{g} \ge g$.

Before we prove the convergence, we state the following preliminary lemma.
\begin{lemma} \label{lem: zero_solution}
The vector $q^*=0$ is an optimal solution to the dual problem (\ref{dual_problem}).
\end{lemma}
\begin{proof}
Since the objective of the primal problem is 0, a feasible solution to the primal is given by any feasible flow allocation $f_{ij}^k$. Since $q=0$ is a feasible dual solution, and any feasible $f_{ij}^k$ together with $q=0$ satisfy the complementary slackness condition (Theorem 4.5 in \cite{bertsimas}), the proof follows.
\end{proof}
This shows that the optimal solution corresponds to queue lengths equal to zero which makes, sense intuitively because any feasible flow allocation in the fluid domain doesn't require queuing. 

Let $G$ be a constant such that it bounds the Euclidean norm of the subgradients of the dual function $D(q)$ for all possible values of $q$, i.e. $G > \|g\|$. From equations (\ref{eq: subgradient1})-(\ref{eq: subgradient2}), we can see that the subgradients are bounded because the flow variables are bounded by link capacities and arrival rates are bounded by assumption. So G is finite. For simplicity we fix $\alpha(t) = 1$ and present the following convergence result.


\begin{theorem}\label{th:2}
Let us approximate the dual variables $q$ with the queue-lengths $\hat{q}$ that evolve according to equations (\ref{eq: overlay_approx})-(\ref{eq: underlay_approx}). Using the dual subgradient descent algorithm with $\alpha(t)=1$, the queue lengths converge to the set $S= \left\{\hat{q}: D(\hat{q}) \le \frac{1}{2}G^2 \right\}$. 
\end{theorem}
\begin{proof}
We will show that $||\hat{q}(t+1) - q^*||^2 < ||\hat{q}(t) - q^*||^2$ when $q(t)$ is outside the set S. Because $q^*=0$ from Lemma \ref{lem: zero_solution}, it suffices to show that $||\hat{q}(t+1)||^2 < ||\hat{q}(t)||^2$. 

We have,
\begin{align*}
||\hat{q}(t+1)||^2 &=  \left\| \hat{q}(t) - \hat{g} \right\|^2
\end{align*}
Since $\hat{g} \ge g$,
\begin{align*}
||\hat{q}(t+1)||^2 &\le \left\| \hat{q}(t) - g \right\|^2 \\
    &= \|\hat{q}(t)\|^2 - 2 \hat{q}(t)^T g + \left\| g\right\|^2
\end{align*}

Our algorithm chooses $g$ to be a subgradient of $D(.)$ at $\hat{q}(t)$. So,
$$ D(x) \ge  D(\hat{q}(t)) + \left(x - \hat{q}(t) \right)^Tg, \forall x \in \mathbb{R}^{m},$$
wehre $m$ is the dimension of $\hat{q}$. Taking $x = 0$, $$D(\hat{q}(t)) \le \hat{q}(t)^T g (f(t)^*)$$
So, $||\hat{q}(t+1)||^2 \le ||\hat{q}(t)||^2 - 2 D(\hat{q}(t)) + G^2.$
Hence when, the $\hat{q}$ is far away from the optimal, specifically when $D(\hat{q}(t)) >  \frac{1}{2}G^2$, it moves towards the optimum in the next time-step. 
\end{proof}

Hence, we will use queue-lengths instead of the dual variables in the implementation of OORP. This will allow us to use the policy presented in Algorithm \ref{alg: OORP} without having to perform the dual variables update.


\subsection{Underlay sources and destinations} \label{sec: underlay_sources}
The problem formulation given in beginning of Section V assumes that all the flows go from one overlay node to another. However, this assumption can be easily removed. Any flow that originates in the underlay be routed over a single path using the underlay routing scheme. Hence these flows can be represented simply as a reduction in the link capacities in the constraints (5) and (6) for the links traversed by this flow. Our algorithm stays unchanged because it is agnostic to the change in the link capacities at the underlay. OORP is also optimal when the underlay is a destination because destination nodes do not perform any routing.

\subsection{Rate control} \label{sec: rate_control}
It is well known that  subgradient descent is a general method to solve convex optimization problems. The dual gradient descent algorithm has been used derive distributed solutions to network utility maximization problems (e.g. \cite{low,lin_shroff}). In the overlay network setting we can get a distributed solution to the utility maximization problem of the form:
$$\max \sum_{k\in K, i\in \mathcal{O}} U^k_i(\lambda_i^k)$$ where $U^k_i(.)$ is concave and strictly increasing.  In this setting, we can assume that there is an infinite backlog at the sources, and the rates $\lambda_i^k$ are chosen to maximize the total network utility. 

We can use the same derivation technique as in section \ref{sec: solution} to obtain a distributed algorithm. The algorithm is very similar to OORP with an added rate controller at each source. The rate control algorithm so obtained is standard, and the joint rate control and routing algorithm can be written as follows: 

\begin{algorithm}[ht]
\caption{Rate control algorithm for the utility maximization problem}
\label{alg: OORP}
At each time-step $t$:
\begin{enumerate}
\item Source node $i\in \mathcal{O}$ for commodity $k$ chooses the rate ${\lambda_i^k}^*$ as follows: 
$${\lambda_i^k}^* = \argmax_{0\le \lambda_i^k \le \mathcal{M}^i_k} {U^k_i(\lambda_i^k)}-q^k_i \lambda^k_i.$$
Here, $\mathcal{M}^i_k$ is a finite upper bound on the rate that the source $i$ can receive.
\item All the overlay nodes use OORP for routing.
\end{enumerate}
\end{algorithm}

\section{Unknown Underlay Queues}
In the previous section we showed that the dual subgradient descent algorithm can be used to compute a feasible rate for each commodity on each link. We also showed that the queue lengths can be used to approximate the subgradient. However, typically legacy devices may not be able to send queue-lengths to the sources. In this section, we will present two approaches to estimate the required queue-length information. The first approach will estimate it using the delay experienced by the packets. The second approach will involve sending probe packets at a certain time intervals that collect the queue-length information in the tunnels.

\subsection{Delay based approaches} \label{sec: delay_approaches}
From equation (\ref{eq: subgradient}) it can be seen that in order to compute the subgradients we only need the total backlog in the tunnel, i.e. we don't need the length of individual queues. A natural approach to estimate the total backlog in a tunnel is by using the time it takes for a packet to traverse it. To implement this method, each tunnel $l$ maintains a delay variable $D_l$. When a packet is sent into a tunnel, the sending node stamps the packet with the current time. When the packet exits the tunnel, the difference between the current time and the time-stamp on the packet is used to update $D_l$. When computing the optimal flow variables, in equation (\ref{eq: subgradient}) of OORP we substitute the sum of the underlay queues-lengths, $\sum_{(a,b)\in l:a \in U}q_{ab}$, with $D_l$.

For this approach, we assume that the underlay uses a FIFO queuing model which is a common forwarding scheme. Hence, this method requires no modification to the underlay. A similar approach has been used by TCP Vegas to solve a network utility maximization problem \cite{low_vegas}.

Although this approach is simple and does not require cooperation from the underlay, the queue-length estimates obtained by this method can be arbitrarily bad. Consider a FIFO queue that is empty at time zero. As shown in Figure \ref{fig: queue_estimation_error}(a), it has an incoming rate of 2 and outgoing capacity of 1. We want to estimate the queue-length at time $t$ by using packet delays. To see the problem with this approach, let us consider a situation when 2 packets arrive at the queue at every time-slot for the first $\tau$ time-slots, and no arrivals happen after that. In this situation, the actual queue length grows at the rate of 1 for the first $\tau$ time-slots, and then it decreases at the rate of 1 packet per time-slot until the queue is empty. On the other hand, the delay increases at the rate of $\frac{1}{2}$, and the last packet (that arrives at the $\tau$th time-slot) sees a delay of 100 because there are 99 packets in the queue at that time. So at time $2\tau$ when the queue is emptied, the packet received will have suffered a delay of $\tau$ time-slots giving a queue-length estimate of $\tau$, whereas the actual queue-length at that time is zero. Furthermore, the estimate stays bad until another arrival happens. This problem is illustrated in Figure \ref{fig: queue_estimation_error}(b). These arbitrarily bad estimates lead to sub-optimality of OORP which we will observe in the simulations. 

A simple way to improve the estimate is to send empty probe packets when real packets are not available for some time period $\mathcal{T}$. A similar approach has been shown to achieve throughput optimality in a special scenario in \cite{georgios_blinkers}. This approach quickly identifies when a queue becomes empty in the absence of new data packets, and the control algorithm can react accordingly. Although this approach corrects the estimate within P time-slots, it can still suffers from the arbitrarily bad estimation errors. As shown in Figure \ref{fig: queue_estimation_error}(c), at time $2\tau$ the estimate is $\tau$ whereas the actual queue-length is zero. Thus, we propose the following approach using explicit probes.

\begin{figure}[t]
\centering
\includegraphics[scale=.7]{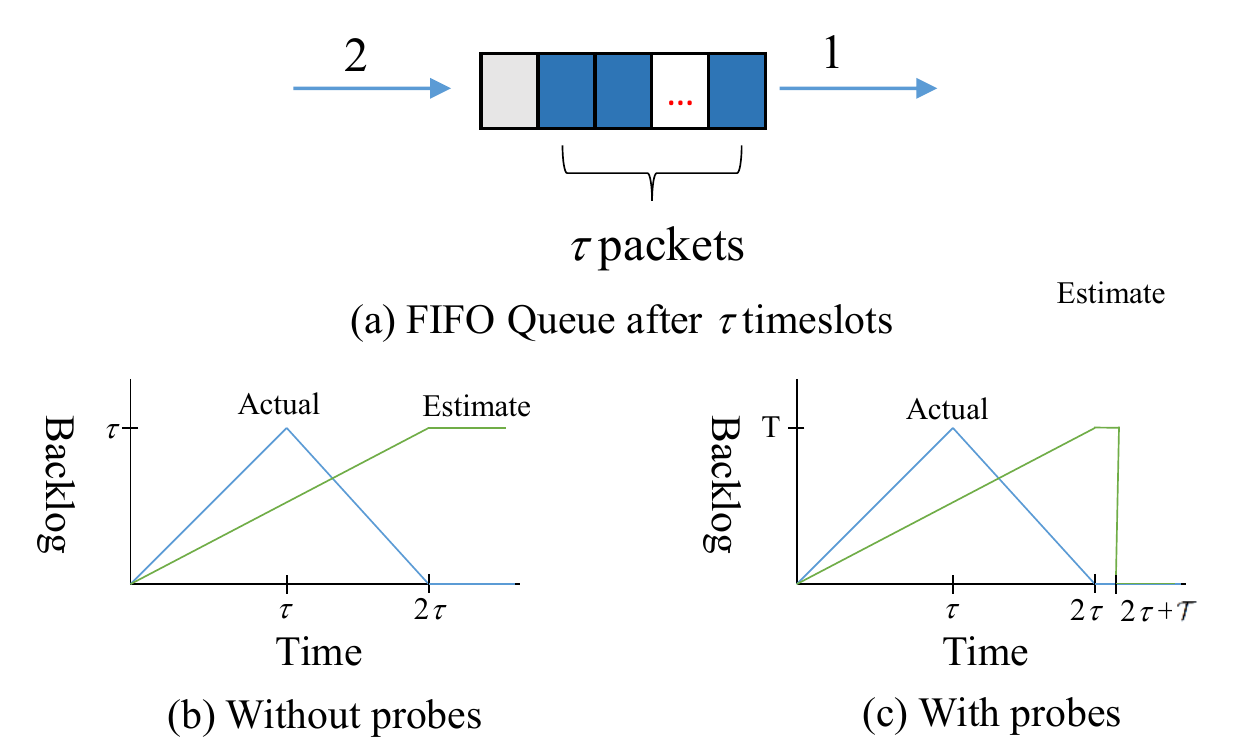}
\caption{The actual queue-length of a single FIFO queue and its estimate calculated using delay. The arrivals happen at the rate of 2 packets per time-slot for the first $\tau$ time-slots, and there are no arrivals after that. Service rate is fixed at 1 packet per time-slot.}
\label{fig: queue_estimation_error}
\end{figure}

\subsection{Priority probe approach}
In this approach, we assume that the underlay nodes are capable of stamping the current queue-lengths into a special type of packets called the probe packets. We also assume that these packets are given higher priority compared to the data packets, and they do not consume link capacity because they are very small in size. These packets are generated by each tunnel at a fixed time-intervals $\mathcal{T}$. When a probe packet exits a tunnel, the sum of the queue-lengths it has collected can be used to compute the optimal flow variables in Algorithm \ref{alg: OORP}.

We can see that this approach results in a much more accurate estimation of the backlog compared to the delay based approach. However, the value of $\mathcal{T}$ can have a significant impact on the performance on the algorithm. We will study its impact in the next section via simulations.

\section{Simulation Results} \label{sec: simulations}
We present several simulation results to evaluate the performance of the optimal overlay routing policy (OORP) given in Algorithm \ref{alg: OORP}. First, we will ascertain that this algorithm is in fact optimal for the network in which the OBP policy of \cite{nathan} was suboptimal. Then we evaluate the effect that different methods of estimating the queue-lengths have on the performance of the algorithm. Next we will study the performance of our algorithm when there is uncontrolled background traffic in the underlay. Finally, we will simulate the rate control algorithm to show that it achieves the maximum throughput.

\subsection{OORP on the counterexample to OBP}
We reconsider the network from Section III, shown in Figure \ref{fig: counter_topology}. The network has three commodities and it can support a maximum arrival rate vector of $\lambda_{\max}^{simA} = [1,1,1]$. We simulate the network under three different policies: the backpressure policy (BP), the overlay backpressure policy (OBP) and our policy, OORP. The simulations are conducted at different loads $\rho = 0.5, 0.55, ..., 1$. For each policy the arrivals are Poisson distributed with rates $\lambda = \rho \lambda_{\max}$. The result of the simulations is given in Figure \ref{fig: counter_simulation}.

\begin{figure}[ht]
\centering
\includegraphics[scale=.7]{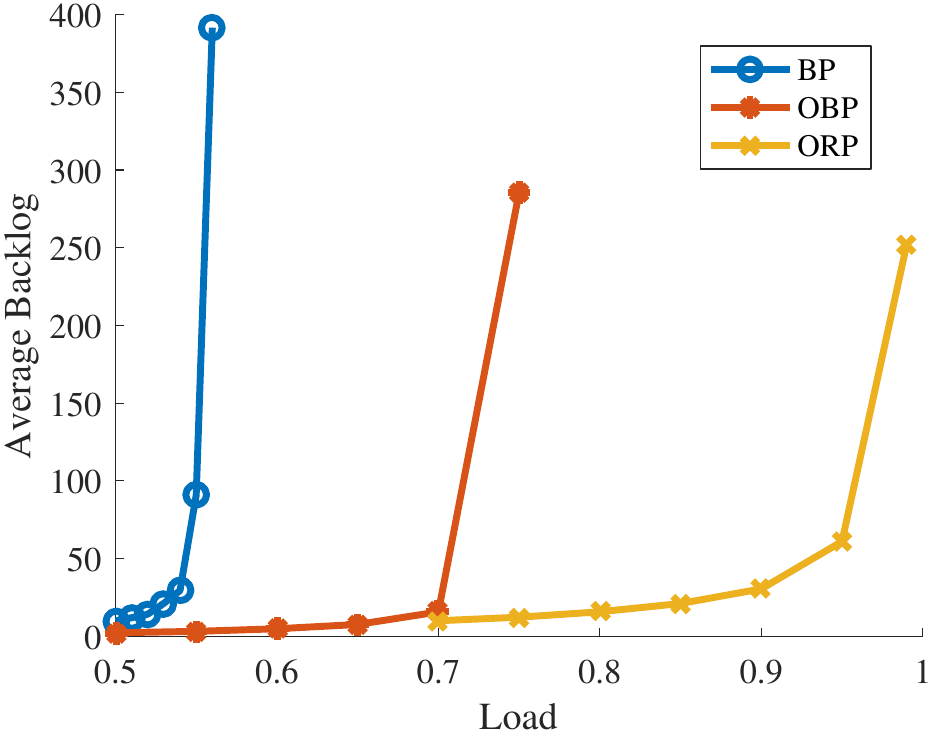}
\caption{Performance of different routing algorithms on the overlay network shown in Figure \ref{fig: counter_topology}.}
\label{fig: counter_simulation}
\end{figure}

The BP algorithm executed on the overlay network does not account for the underlay nodes and simply views a tunnel $l$ as a link between two overlay nodes with capacity $c_{l_1 l_2}$. In the plot we can see that this algorithm becomes unstable around the load of $0.56$. This is as expected because for each commodity the backpressure policy uses both tunnels equally since they have equal weights. The end node of both the tunnels is a destination, which has zero backlog, and BP does not account for the backlog in the underlay. So the weigh for each tunnel of commodity $k$ is equal to $Q_{s_k}^k$. When the algorithm uses the longer tunnel, the network becomes unstable for relatively low load.

The plot also shows that OBP is suboptimal and OORP achieves maximum throughput. We discussed the suboptimality of OBP in Section III. The main reason was that the OBP policy could not avoid using the longer tunnel which gave raise to a cycle of increased congestion. But in OORP, when the underlay queues-lengths are positive, the shorter tunnels have a higher weight than the longer tunnels. For example, for commodity 1, the weight of the shorter tunnel $(s_1,1,2,d_1)$ is $Q^1_{s_1} - Q_{12}-Q_{2d_1}$ and the weight of the longer tunnel $(s_1,3,1,2,d_1)$ is $Q^1_{s_1}-Q_{31} - Q_{12}-Q_{2d_1}$. So when the underlay queues are large, the longer tunnel needs a lot more backlog at the source than the shorter tunnel for its weight to be positive. This causes OORP to avoid using the longer tunnels when the network is congested.

\subsection{Estimated Tunnel Backlog}
We consider the network given in Figure \ref{fig: final_simulation} to observe the effect of estimating the backlog in the tunnels.  In this network, all the links are bidirectional, composed of two unidirectional links. The links between an overlay and an underlay node have capacity 2 in each direction. All other links have unit capacity in both directions. We will simulate the network with two commodities. The first commodity is defined by the source-destination pair (1,3) and the second is defined by (2,4). For these commodities the network supports a max-flow vector of $\lambda_{\max}^{simB} = [2,2]$. The simulation is performed at various load levels and the arrivals are Poisson distributed. 

The underlay uses the shortest path routing hence creating a large number of available tunnels. Node 1 can send packets to node 3 directly via node 7 or 10 using the tunnels (1,7,5,6,3) and (1,10,7,5,6,3) respectively. However, these tunnels overlap, hence there is no benefit in using both of them. To achieve the throughput of two, node 1 must send its traffic through node 2 and have it forward it to node 3. Similarly node 2 must send some of its traffic through node 1 in order to achieve high throughput. Observing the organization of the tunnels in the network, we can see that using the wrong tunnel might cause the network to lose throughput. In addition, the tunnels form cycles in the overlay topology. These features make this topology challenging for a routing algorithm to achieve the optimal throughput.

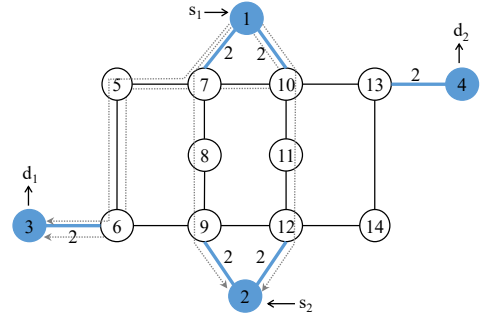
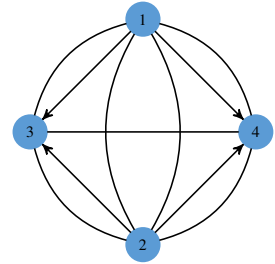
\begin{figure}[ht]
\centering
\begin{tabular}{c}
\subfloat[Physical network topology with overlay (blue) and underlay (white) nodes. The underlay network uses shortest path routing creating a total of eighteen tunnels between the overlay nodes. The dotted lines show the tunnels from node 1 to nodes 2 and 3.]{
    \makebox[7cm][c]{
        \includegraphics[scale=.5]{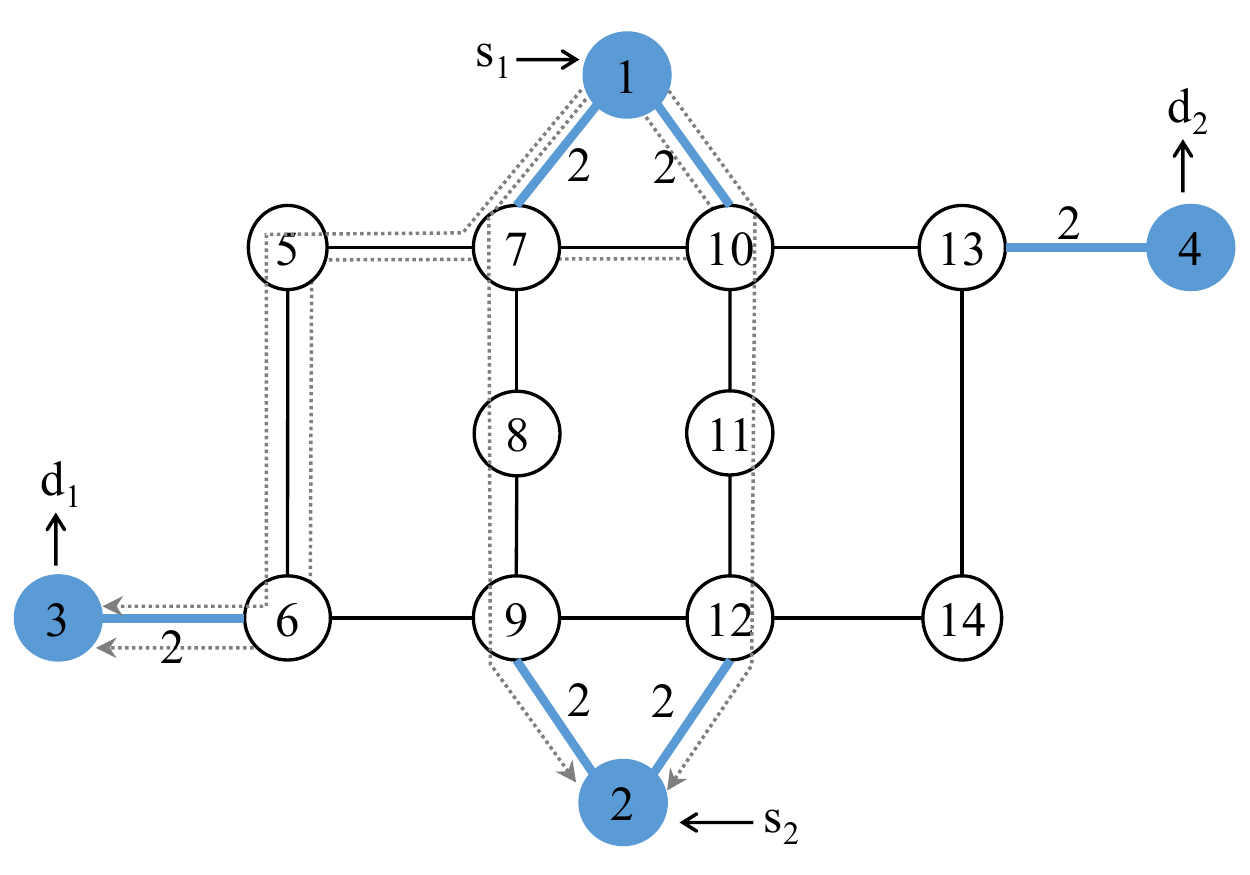}
    }
}\\

\subfloat[Corresponding overlay network showing all eighteen tunnels. Each link in this graph represents a tunnel; an undirected link represents two tunnels, one in each direction. E.g. the two links going from node 1 to 3 represent the two tunnels from node 1 to 3:  (1,7,5,6,3) and (1,10,7,5,6,3).]{
\makebox[7cm][c]{
    \begin{tikzpicture}[-, >=stealth', auto, semithick, node distance=3cm]
    
    \tikzstyle{every state}=[fill=mycolor, draw=mycolor, thick,text=black,scale=.5]
    
    \node[state, text=black] (3) {\large 3};
    \node [state, text=black] (1)[above of=3,right of=3] {\large 1};
    \node[state, text=black] (4)[right of=3, right of=3] {\large 4};
    \node[state, text=black] (2)[below of=3, right of=3] {\large 2};

    \path (3) edge[bend left](1);
    \path [<-] (3)  edge[right](1);
    \path [<-] (3) edge[left] (2);
    \path (3) edge[bend right](2);
    \path [<-] (4) edge[left] (1);
    \path (4) edge[bend right](1);
    \path (4) edge[bend left] (2);
    \path [<-] (4) edge[right](2);
    \path (1) edge[bend left] (2);
    \path (1) edge[bend right](2);
    \path (3) edge (4);
    \end{tikzpicture}
}
}
\end{tabular}
\caption{Physical and overlay network topology for simulations in Sections \ref{sec: simulations} B and \ref{sec: simulations} C. }
\label{fig: final_simulation}
\end{figure}

The result of the simulations under different load levels is given in Figure \ref{fig: estimation_simulation}. We can see that the delay approach, which uses packet delay as an estimate of the tunnel backlog, does not provide optimal throughput. Although probing the delay in the network with control packets improves the performance, it is still suboptimal. This happens because when the backlog is large the estimation error of this approach can also be large as described in Section \ref{sec: delay_approaches}.

We can also see that the probing approach achieves optimal throughput, and its performance is close to that of using the actual queue-lengths. The estimates obtained by this approach are much more accurate than those from the delay approach because they are not affected by the amount of backlog in the network. When $\mathcal{T}$ is increased, the stale estimate is used for a longer time period, so the performance of the algorithm degrades.

\begin{figure}[ht]
\centering
\includegraphics[scale=.7]{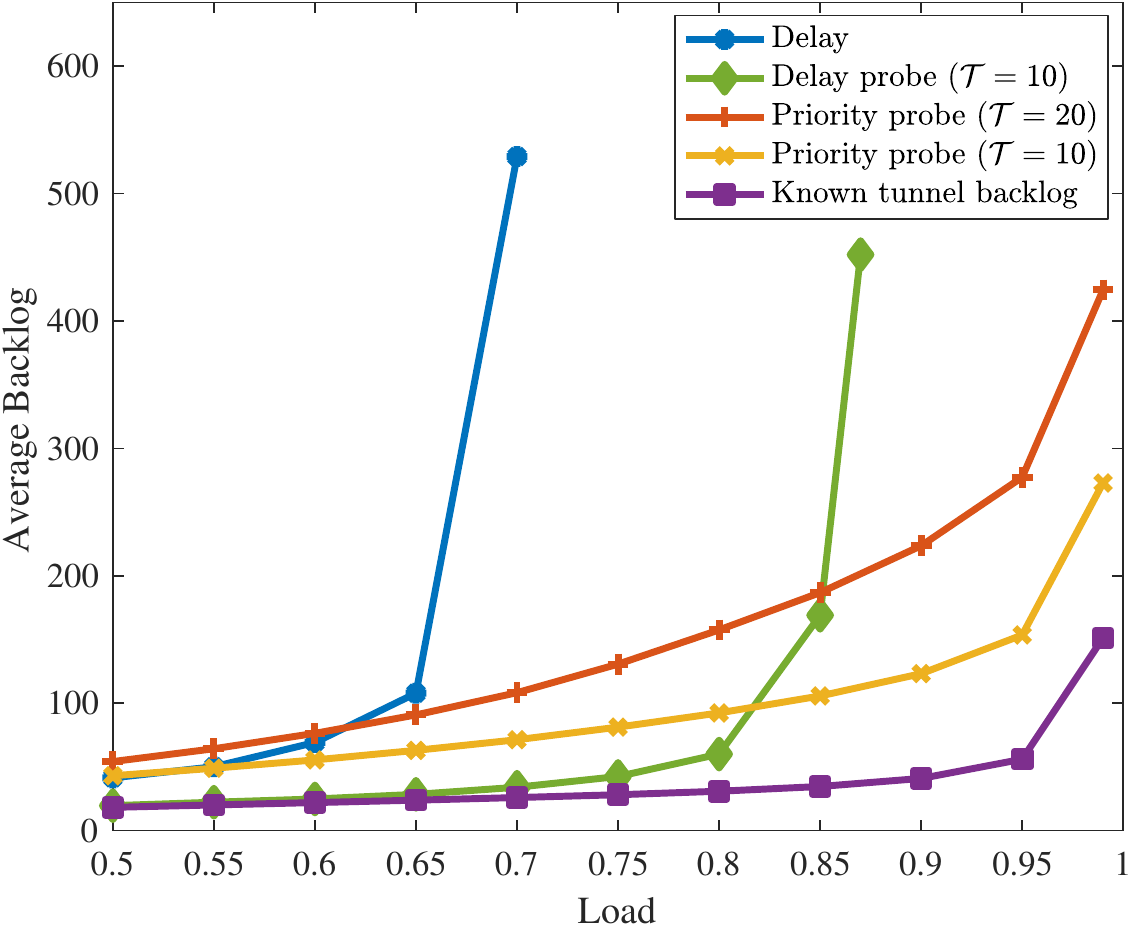}
\caption{ Performance of OORP under different measures of tunnel backlog.}
\label{fig: estimation_simulation}
\end{figure}

\subsection{Background traffic} \label{sec: background_sim}
So far we have assumed that all the traffic in the network belongs to the overlay network. However, in real networks the underlay can be routing other traffic not generated by the overlay nodes. Next, we will study the performance of our algorithm under such traffic. We expect the OORP to be throughput optimal under stable background traffic in the underlay because such traffic can be thought of as a reduction in the link capacities in inequalities (\ref{eq: tunnel_constraint_fluid}) as described in Section \ref{sec: underlay_sources}.

We again consider the network from Figure \ref{fig: final_simulation} with two commodities (1,3) and (2,4). We inject two flows of background traffic: first going from node 7 to 6 along the path (7,5,6) with the arrival rate of 0.5, and second going from 8 to 14 along the path (8, 9, 12, 14) with the arrival rate of 0.2. The arrivals happen according to the Poisson process. Note that the first background flow blocks commodity 1's tunnel (1,7,4,5) and the second flow blocks commodity 2's tunnel (2,12,14,13,4); both the tunnels are essential for achieving the max flow vector $\lambda^{simB}_{\max}$. This reduces the maximum supportable arrival rates for the two commodities to $\lambda^{simC}_{\max}=[1.5, 1.8]$. Figure \ref{fig: background_traffic_simulation} shows the result of the simulation. We can see that all the approaches except for the delay based approaches achieve the maximum throughput.

\begin{figure}[ht]
\centering
\includegraphics[scale=.7]{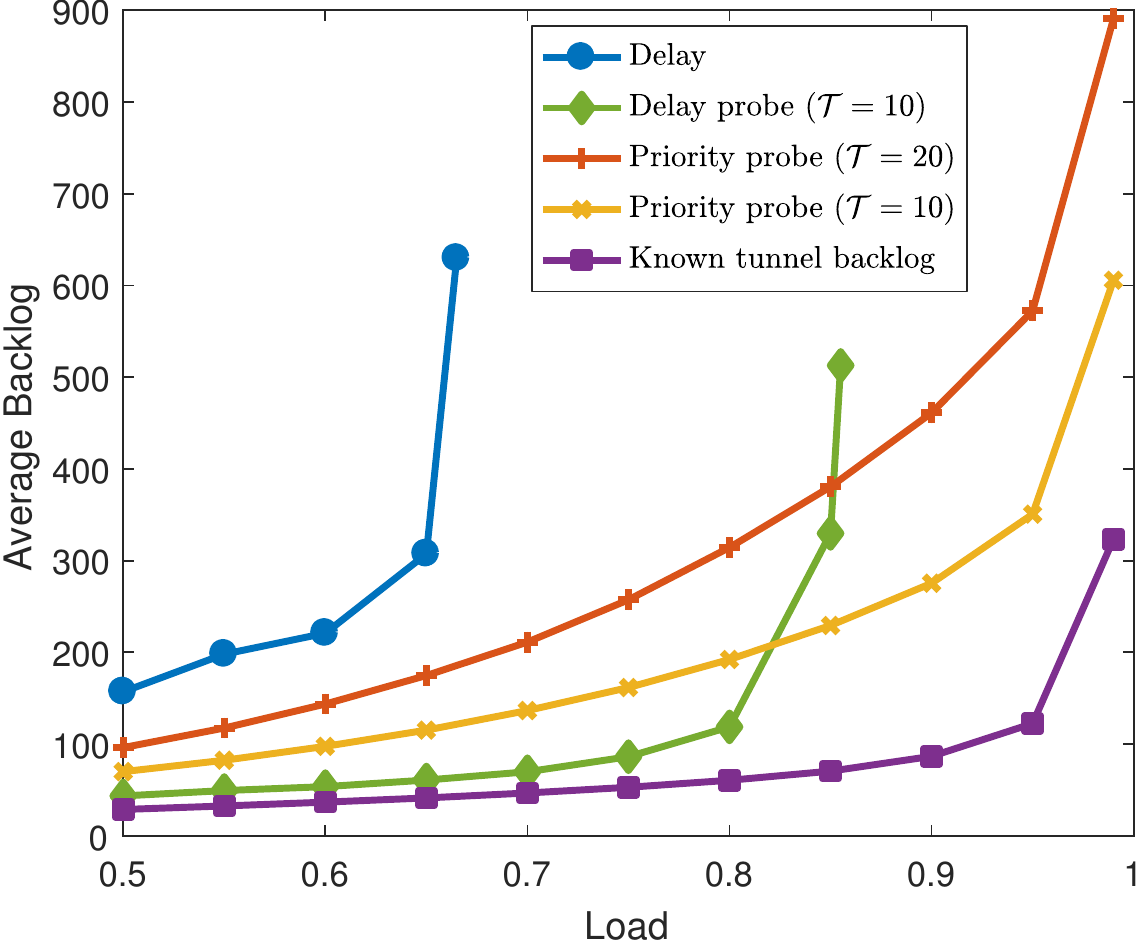}
\caption{Performance of OORP in a network with background traffic.}
\label{fig: background_traffic_simulation}
\end{figure}


\subsection{Rate control}
To observe our rate controller at work, we consider the network from Section \ref{sec: background_sim} with a minor modification as shown in Figure \ref{fig: rate_control_topology}. We add a new overlay node 15 and a new commodity (15, 14) to the network. Node 15 connects to node 5 with a directed link (15, 4) which has a unit capacity. We constrain the third commodity to use the tunnel provided by the shortest path (5, 6, 9, 12, 14). For all three sources, the utility function is chosen to be $20\log(\lambda)$  and $M_i^k=20$. Note that the addition of the third commodity makes the simulation more challenging because the rate that maximizes total throughput is not the same as the rates that maximizes utility. We assume that the backlog information of each tunnel is available to the overlay nodes instantaneously. 

\begin{figure}[ht]
\centering
\includegraphics[scale=.5]{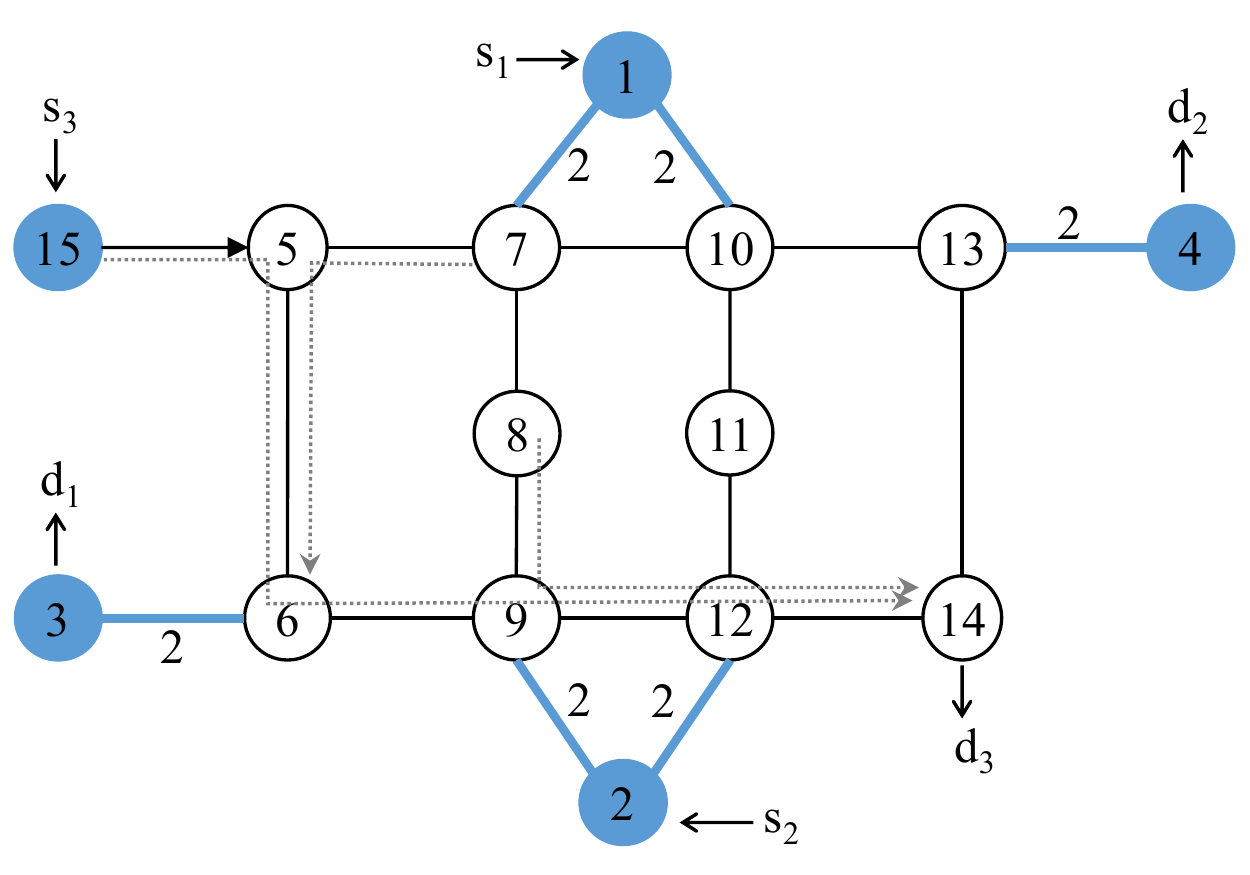}
\caption{Topology for the rate control experiment. The dotted lines show the tunnel assigned to the third commodity and the background traffic.}
\label{fig: rate_control_topology}
\end{figure}

Constrained by the link capacities and the background traffic, the maximum throughputs for the commodities 1,2, and 3 are 1.5, 1.8, and 1 respectively. However, the third commodity interferes with both commodities 1 and 2, hence the throughput of [1.5, 1.8, 1] is not achievable. From the plot in Figure \ref{fig: throughput_plot} we can see that the throughput vector converges to [1, 1.3, 0.5] which maximizes the total utility. We can see that this throughput vector has a smaller sum than the sum of the maximum throughput supported in Section \ref{sec: background_sim}. That is, the network could have supported higher throughput by giving zero throughput to the third commodity, but that would have decreased the utility of the network.

\begin{figure}[ht]
\centering
\includegraphics[scale=.75]{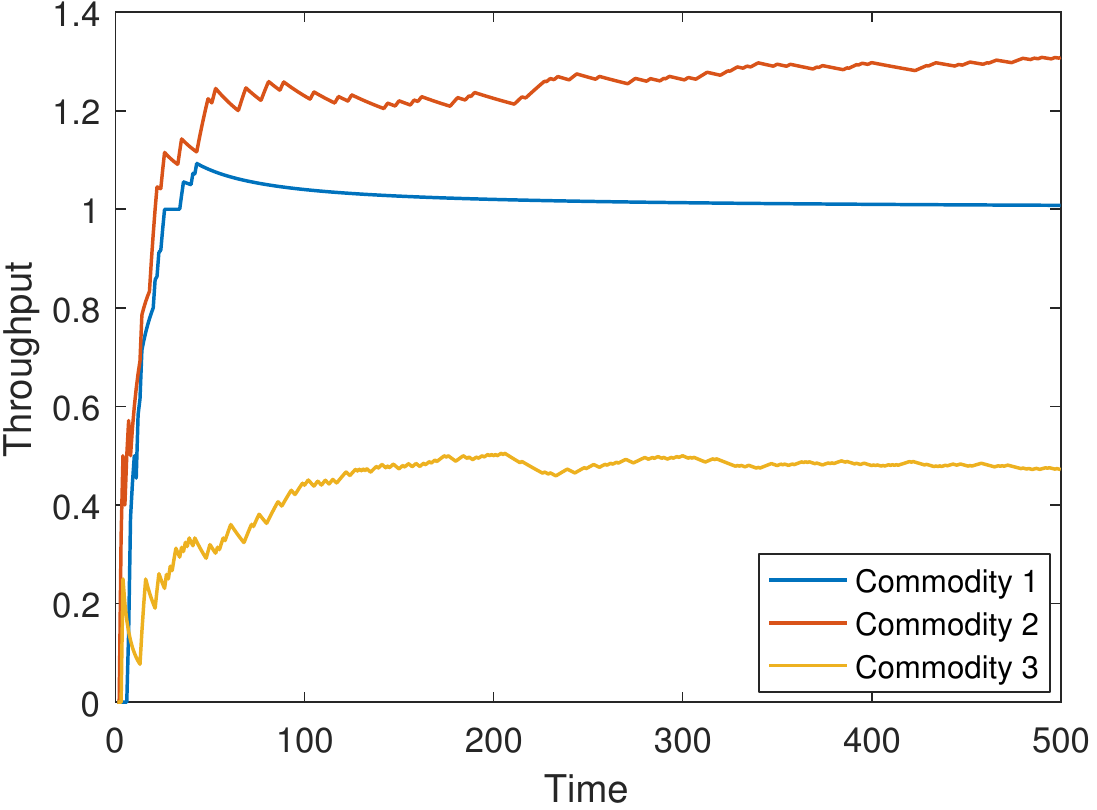}
\caption{Throughput achieved by the rate control algorithm with OORP converges to the rate that maximizes utility.}
\label{fig: throughput_plot}
\end{figure}

\section{Conclusion}
We showed that the existing algorithms for routing traffic in an overlay network are suboptimal, and developed a throughput optimal policy called the Optimal Overlay Routing Policy (OORP). This policy is distributed and can also be used with a rate controller to maximize network utility. Our algorithm requires the knowledge of congestion at the underlay, which might not be available to the overlay nodes. Hence we proposed different approaches to estimating underlay congestion. Simulations results show that OORP outperforms OBP and that estimating congestion using probing mechanism is effective. Future research will include obtaining better estimates for the congestion in the tunnels with minimum support from the underlay nodes and reducing the delay experienced by the packets in the network.

 \appendix[Proof of Theorem \ref{thm: centralized}] \label{app: centralized_proof}
 \subsection{Stationary policy $\pi$ }
 In order to prove the stability of the centralized policy, we need a stationary policy that stabilizes the network. 
 For any arrival rates ${\bf \lambda}$ such that ${\bf \lambda} + {\bf \epsilon} \in \Lambda$, $\epsilon>0$, we know that there exist a feasible flow allocation vectors $(f^k_{l})_{l\in L, k \in K}$ such that for any overlay node $n$,
 \begin{equation} \label{eq: conservation}
 \sum_{l\in L: l_1 = n} f^k_{l} - \sum_{l\in L: l_{|l|} = n} f^k_{l} = \lambda^k_n + \epsilon.
 \end{equation}
 This vector can be obtained by solving the multi-commodity flow problem. We assume that these flow variables can be closely approximated by rational numbers. So there exists integers  $p^k_{l}$ and $q$ such that $f^k_{l} = p^k_{l}/q$. The time-slot are divided into $T$ slot long frames. 

 The policy $\pi$ simply sends $p^k_{l}$ amount of commodity $k$ packets every $q$ time-slots in each frame. Because the underlay is using a universally stable forwarding scheme and the burstiness constraints are satisfied, the underlay queues are deterministically bounded by a constant $B$ \cite{andrews}. Also note that all the capacity constraints are satisfied every $q$ time-slots. Hence, $\pi$ stabilizes $\lambda$.

 Let $F^k_{l}(t+\tau,\pi)$ represent the number of packets sent into tunnel $l$ by node $l_1$ at time $t+\tau$ under policy $\pi$. Let $\bar{F}^k_{l}(t+\tau,\pi)$ represent the number of packets that are received by node $l_{|l|}$ at time $t$ from tunnel $l$ under policy $\pi$. Note that $F^k_{l}(t+\tau,\pi) = \bar{F}^k_{l}(t+\tau,\pi)$ only if the tunnel $l$ is a direct link between two overlay nodes. If a tunnel passes through the underlay, it can take a bounded amount of time for the packets to exit the tunnel. Now, we prove the following lemma that will be used in proving the theorem.

 \begin{lemma}\label{lem: rand_drift}
 For the proposed randomized policy $\pi$
 \begin{align}
  &\mathbb{E}\left[   \sum_{l\in L: l_1 = n}  \sum_{\tau=0}^{T-1}F^k_{l}(t+\tau,\pi) -  \right. \nonumber \\
  & \left.\left. \sum_{{l\in L: l_{|l|} = n} }\sum_{\tau=0}^{T-1} \bar{F}^k_{l}(t+\tau,\pi) \right|{\bf Q}(t)\right]  \ge T(\lambda_n^k + \epsilon) - B, \forall n. \nonumber
 \end{align}
 \end{lemma}

 \begin{proof}
 \begin{align}
  &\mathbb{E}\left[  \sum_{l\in L: l_1 = n}  \sum_{\tau=0}^{T-1}F^k_{l}(t+\tau,\pi) - \right. \nonumber \\
 &\qquad \qquad \left. \left. \sum_{{l\in L: l_{|l|} = n} }\sum_{\tau=0}^{T-1} \bar{F}^k_{l}(t+\tau,\pi)  \right|{\bf Q}(t)\right]  \nonumber \\
 &=\mathbb{E}\left[ \sum_{l\in L: l_1 = n}  \sum_{\tau=0}^{T-1}F^k_{l}(t+\tau,\pi) -  \right. \nonumber\\
 & \qquad \qquad \left. \sum_{{l\in L: l_{|l|} = n} }\sum_{\tau=0}^{T-1} \bar{F}^k_{l}(t+\tau,\pi) \right] \\
 &\ge\mathbb{E}\left[ \sum_{l\in L: l_1 = n}  \sum_{\tau=0}^{T-1}F^k_{l}(t+\tau,\pi) - \right. \nonumber\\
 & \qquad \qquad \left. \sum_{{l\in L: l_{|l|} = n} }\sum_{\tau=0}^{T-1} F^k_{l}(t+\tau,\pi) - B  \right] \label{eq: bound}\\
 &= T\sum_{l\in L: l_1 = n} \frac{p^k_{l}}{q} - T\sum_{l\in L: l_{|l|} = n} \frac{p^k_{l}}{q} - B\\
 &= T\sum_{l\in L: l_1 = n} f^k_{l} - T\sum_{l\in L: l_{|l|} = n} f^k_{l} - B\\
 &= T(\lambda_n^k + \epsilon) - B \nonumber
 \end{align}
 Here $B$ is a finite constant representing the maximum amount of backlog in the underlay network. We use this constant to obtain inequality (\ref{eq: bound}) and equation (\ref{eq: conservation}) to obtain the last equality.
 \end{proof}

 \subsection{Analysis of TBP}
 We know that the underlay queues are stable because the traffic injected into the tunnels satisfy the burstiness constraints and the underlay employs a universally stable forwarding policy \cite{andrews}. Next we prove the stability of overlay queues.

 The queue evolution of an overlay node $n$ can be written as:
 \begin{align*}
 Q^k_n(t+1) &= \left[ Q^k_n(t) - \sum_{l: l_1 = n} F^k_{l}(t)  +   \sum_{l:l_{|l|}=n} \bar{F}^k_{l}(t) + A^k_n (t)\right]^+ \\
 &\le \left[ Q^k_n(t) - \sum_{l: l_1 = n} F^k_{l}(t) \right]^+ +   \sum_{l:l_{|l|}=n} \bar{F}^k_{l}(t) + A^k_n (t) 
 \end{align*}
 Here $F^k_{l}(t)$ represents the amount of packets injected into the tunnel $l$ at time $t$, and $\bar{F}^k_{l}(t)$ represents the number of packets that exit tunnel $l$ at time $t$.

 Then the queue length after $T$ slots can be bounded as follows
 \begin{align}
 Q^k_n(t+T) \le &\left[ Q^k_n(t) - \sum_{l: l_1 = n} \sum_{\tau=0}^{T-1}F^k_{l}(t+\tau)  \right]^+ +\nonumber\\
 		& \sum_{l:l_{|l|}=n} \sum_{\tau=0}^{T-1} \bar{F}^k_{l}(t+\tau) + \sum_{\tau=0}^{T-1}A^k_n (t+\tau) \nonumber\\
 		\le &\left[ Q^k_n(t) - \sum_{l: l_1 = n} \sum_{\tau=0}^{T-1}F^k_{l}(t+\tau)  \right]^+ +\nonumber\\
 		& \sum_{l:l_{|l|}=n} \sum_{\tau=0}^{T-1} F^k_{l}(t+\tau) + \sum_{\tau=0}^{T-1}A^k_n (t+\tau) + B \nonumber
 \end{align}
 The first inequality comes from considering all the arrivals and departures in a T-slot interval in a single slot. We get the second inequality by bounding the backlog in the underlay nodes by a constant $B$.

 Now to prove the theorem, consider the quadratic Lyapunov function $$L({\bf Q}(t)) = \sum_{k,n} \left( Q^k_n(t) \right)^2.$$ The {\em T-slot drift} is given by:
 \begin{align}
 \Delta_T = &\mathbb{E}\left[L({\bf Q}(t+T)) - L({\bf Q}(t)) | {\bf Q}(t) \right] \nonumber\\
 	\le&T^2K + \sum_{k,n} Q^k_n(t) \left(T \lambda^k_n + B\right) + \sum_{k,n} Q^k_n   \nonumber\\
 	&\left. \mathbb{E} \left[  \sum_{\tau=0}^{T-1} \sum_{l:l_1=n} F^k_{l}(t+\tau) - \sum_{\tau=0}^{T-1} \sum_{l:l_{|l|}=n} F^k_{l}(t+\tau) \right|{\bf Q}(t)\right] \label{eq: drift1}\\
 = &T^2K + \sum_{k,n} Q^k_n(t) \left(T \lambda^k_n + B\right)  \nonumber \\
 	&- \mathbb{E}\left[\left.  \sum_{\tau=0}^{T-1} \sum_{k,l} F^k_{l}(t+\tau)  \left( Q^k_{l_1}(t) - Q^k_{l_{|l|}}(t) \right) \right| {\bf Q}(t)  \right] \label{eq: drift2}
 \end{align}

 The TBP policy minimizes the right hand side of inequality (\ref{eq: drift1}) at every time-slot.  Hence it also minimizes the right hand side of inequality (\ref{eq: drift2}). So we can bound the drift by the rate variables chosen by the stationary policy.

 \begin{align}
 \Delta_T \le&T^2K + \sum_{k,n} Q^k_n(t) \left(T \lambda^k_n + B\right) + \sum_{k,n} Q^k_n   \nonumber\\
 	& \mathbb{E} \left[  \sum_{\tau=0}^{T-1} \sum_{l:l_1=n} F^k_{l}(t+\tau, \pi) \right.\nonumber \\
 	&\left.\left.- \sum_{\tau=0}^{T-1} \sum_{l:l_{|l|}=n} F^k_{l}(t+\tau, \pi) \right|{\bf Q}(t)\right] \\
 	\le& T^2K + \sum_{k,n} Q^k_n(t) \left(T \lambda^k_n + B \right) - \sum_{k,n} Q^k_n (T\lambda_n^k \nonumber \\
 	&+ T\epsilon-B) \label{eq: rand_drift}\\
 \Delta^{TBP}_T \le &T^2K - \sum_{k,n}Q^k_n(t) (T\epsilon-2B)
 \end{align}
 We use Lemma \ref{lem: rand_drift} to obtain (\ref{eq: rand_drift}). The drift is negative when $T>2B/\epsilon$ and the queues are large. From \cite{neely_book} we know that the overlay queues are strongly stable.

\end{document}